\documentclass[%
   pre,
   superscriptaddress,
   twocolumn
]{revtex4-1}
\usepackage{graphicx}				
\usepackage{amssymb}
\usepackage{amsmath}
\usepackage{xcolor}
\usepackage[hidelinks]{hyperref}
\usepackage{blkarray}
\usepackage{accents}
\usepackage{versions}
\excludeversion{solution}
\usepackage{tikz}
\usetikzlibrary{automata,positioning}
\graphicspath{{../figures/}}

\newcommand{\ie}{{\em i.e., }}
\newcommand{\eg}{{\em e.g., }}

\newtheorem{theorem}{Theorem}
\newtheorem{proof}{Proof}
\usepackage{algcompatible,newfloat}
\DeclareFloatingEnvironment[
    fileext=loa,
    listname=List of Algorithms,
    name=Algorithm,
    placement=tbhp,
]{new_algorithm}
\newcommand{\ER}{Erd{\"o}s-R{\'e}nyi }

\newcommand{\titlestr}{Generating Connected Random Graphs}

\begin{document}

\title{\titlestr}
\author{Caitlin Gray}
\email{caitlin.gray@adelaide.edu.au}
\affiliation{School of Mathematical Sciences, University of Adelaide}
\affiliation{ARC Centre of Excellence for Mathematical \& Statistical  Frontiers}
\author{Lewis Mitchell}
\affiliation{School of Mathematical Sciences, University of Adelaide}
\affiliation{ARC Centre of Excellence for Mathematical \& Statistical  Frontiers}
\affiliation{DatatoDecisions CRC: Stream Lead}
\author{Matthew Roughan}
\affiliation{School of Mathematical Sciences, University of Adelaide}
\affiliation{ARC Centre of Excellence for Mathematical \& Statistical  Frontiers}

\begin{abstract}
Sampling random graphs is essential in many applications, and often algorithms use Markov chain Monte Carlo methods to sample uniformly from the space of graphs. However, often there is a need to sample graphs with some property that we are unable, or it is too inefficient, to sample using standard approaches. In this paper we are interested in sampling graphs from a conditional ensemble of the underlying graph model. We present an algorithm to generate samples from an ensemble of connected random graphs using a Metropolis-Hastings framework. The algorithm extends to a general framework for sampling from a known distribution of graphs, conditioned on a desired property. We demonstrate the method to generate connected spatially embedded random graphs, specifically the well known Waxman network, and illustrate the convergence and practicalities of the algorithm.
\end{abstract}
\maketitle
\section{Introduction}
Random graphs are commonly used as underlying models in many fields, such as computer networking, biology, social sciences and physics \citep{PayneDodds-09,Bascompte-07,Crucitti-04,Gray-18,Castellano-15}. The ability to generate random graphs with desired properties is crucial, as they may be used in conjunction with complex models, for instance a routing protocol in computer networking \citep{Waxman-88}. 


Real-world networks come with countless properties that one may consider modelling, \eg degree distributions, clustering levels etc. Most random graph models focus on one of these properties to model an observed network. However, many current methods for generating random graphs result in networks with some undesirable properties for a particular applications. 

For instance,
\begin{itemize}
\item the graphs may not be connected, \eg the Gilbert-\ER model 
or spatial Waxman graph \citep{Waxman-88}; or 
\item the graphs may not be \textit{simple}, \ie 
they might have multi-edges or self-loops, \eg the configuration model.
\end{itemize}

While one might argue that this is a modelling problem, there are nevertheless many instances in the literature where a model matches enough properties of the real networks in question that it is useful, except for one deficiency such as noted above.

Examples include:
\begin{itemize}
\item using the Waxman graph to model physical networks that are inherently 
connected, \eg router networks; and
\item using the configuration model that generates graphs with self-loops and multi-edges to model simple networks.
\end{itemize}

Generating connected graphs with a given degree sequence has been discussed at length in the literature using Markov chain Monte Carlo (MCMC) methods \citep{Viger-05,Rechner-17,Nishimura-18,Ying-09}. The existing MCMC algorithms use `edge swaps' to give a uniform sample over the graph space. While this may be useful when requiring only a graph with the desired property, the natural question remains of how to sample graphs while ensuring we maintain the conditional ensemble of the underlying graph model. This is essential in many applications; for example, when estimating parameters, or in applications of Approximate Bayesian Computation where the ensemble encompasses prior knowledge of the system.

We present an algorithm to produce random graphs from a known ensemble conditioned on an extra desired property of the network. Our algorithm uses MCMC methods to sample from the ensemble of interest. In particular, we focus here on generating connected networks. We show the algorithm samples graphs from the desired distribution and demonstrate the algorithm on spatially embedded random networks (SERNs), in particular the Waxman random graph \citep{Waxman-88}. We show that the algorithm is $\mathcal{O}(K)$ for $K$ iterations, and show convergence scales like $\mathcal{O}(N^2)$ in the number of nodes in the graph.

The algorithm not only has practical applications in that one can generate connected graphs for use in various applications, but also, such a simulation algorithm could be used to estimate the probability of such graphs in an ensemble.

\section{Background}

\subsection{Mathematical formalities}

A graph (or network), $G=(V, E )$, consists of a set of $N$ nodes, which, without loss of
generality, we label ${V}=\{1,2,\ldots, N\}$. The graph has edges (or
links) ${E} \subset {V} \times {V}$. We are primarily
concerned here with undirected graphs (though much work on random
graphs is easy to generalise to directed graphs).

We say that a link exists between two nodes $i$ and $j$ if $(i,j) \in {E}$. We say that they are {\em connected} if a path (a sequence of edges) exists between the two nodes. The graph is connected if all pairs of nodes $(i,j)$ are connected.

The well-known Gilbert-\ER (GER) random graph, $G_{n,p}$
of $n$ nodes is constructed by assigning each edge $(i,j)$ to be in
${E}$ independently, with fixed probability $p$ ~\citep{ER-59,Gilbert-59}. 

Spatially embedded random networks (SERNs) stem from the notion that often longer links are more expensive to build or maintain. Therefore, often real world networks display spatial structure, and are used in social and epidemiological modelling \citep{Miller-18, Bringmann-18}. Formally, we create a SERN by placing $N$ nodes uniformly at random within some
defined region $R$ of a metric space $\Omega$ with distance metric
$d(x,y)$. Each pair of nodes is made adjacent independently, with
probability $p_{ij}$, which is a function of $d(x_i,x_j)$. In the Waxman case, 
\begin{equation}
p_{ij}= qe^{-sd_{ij}},
\label{eq:waxman_prob}
\end{equation} 
for $q \in (0, 1]$, $s \ge 0$, and the Euclidean distance $d_{ij}$. The parameter $s$ controls the extent to which spatial structure is incorporated into the graph. Note that when $s=0$ we recover the GER random graph, with edge probability $q$. In general, the $q$ value controls the overall edge density in the graph.
Note that the parametrisation in~(\ref{eq:waxman_prob}) differs from much
of the literature on Waxman graphs. We chose to do this as
unfortunately, the parameters $(\alpha, \beta)$ used traditionally
have become confused by frequent reversal. 

The basic properties of the Waxman graph can be derived. For instance,
it is shown \citep{Roughan-15} that the average node degree is given by
\begin{equation}
  \label{eq:node_degree}
  \bar{z} = (n-1) q \tilde{G}(s),
\end{equation}
where $\tilde{G}(s)$ is the Laplace transform of the probability density function between a pair of random points (the Line-Picking Problem), see references for further details \citep{Roughan-15, Ghosh-51}. The Waxman is just one example of a SERN, and we use it here to provide a simple and clear example. Results generalise to other SERNs.

\subsection{Markov chain Monte Carlo}
Markov chain Monte Carlo (MCMC) methods are widely used to sample from complex probability distributions that are difficult to generate directly. These approaches generate Markov chains that converge to the distribution of interest. 

Specifically, we use the Metropolis-Hastings (M-H) algorithm \citep{Metropolis-49, Hastings-70}, given in Algorithm \ref{MH_algorithm}, to draw samples from our distribution of interest, namely, the distribution of networks with our desired property. 
\begin{new_algorithm}
  \begin{algorithmic}[1]
  \STATE Set $\theta^{(0)}$
  \FOR {$t=1...K$}
    \STATE Generate $\theta' \sim Q(\theta' | \theta^{(t-1)})$
	\STATE Take 
			$\theta^{(t)}  = \left\{
       \begin{array}{ll}
          \theta', & \mbox{ with probability } \alpha \\
          \theta^{(t-1)},    & \mbox{ with probabiltiy } 1-\alpha.
       \end{array}
     \right.$  \
	
where $\alpha = \text{min}\left(1,\frac{\pi(\theta')Q(\theta \vert \theta')}{\pi(\theta) Q(\theta' \vert \theta)}\right) $
    	\ENDFOR
    \caption{General Metropolis-Hastings algorithm \citep{Robert-04}.}
    \label{MH_algorithm}
  \end{algorithmic}
\end{new_algorithm}

Consider the target distribution $\pi(\theta)$ we wish to sample from. We use the M-H algorithm to create a Markov chain $\theta^{(1)}, \theta^{(2)}, \cdots$. To do so, we choose a proposal distribution $Q(\theta' \vert \theta)$ to propose the next candidate $\theta'$ from the current state $\theta$. The proposal distribution must be able to explore the entire space in a finite number of steps \citep{Robert-04}.

The proposed parameter value $\theta'$ is accepted with some probability given by, in the case of M-H, the acceptance probability
\begin{equation*}
\alpha = \text{min}\left(1,\frac{\pi(\theta')Q(\theta \vert \theta')}{\pi(\theta) Q(\theta' \vert \theta)}\right).
\end{equation*}
If the proposal distribution is symmetric then
\begin{equation*}
\alpha = \text{min}\left(1,\frac{\pi(\theta')}{\pi(\theta)}\right).
\end{equation*}
The chain is generated from the proposed parameter $\theta'$ as follows
\begin{equation*}
\theta^{(t+1)}
     = \left\{
       \begin{array}{ll}
          \theta', & \mbox{ if accepted}, \\
          \theta^{(t)},    & \mbox{ otherwise},
       \end{array}
     \right.
\end{equation*}
where $\theta'$ is generated from $Q(\theta' \vert \theta^{(t)})$.

Markov chain traversals of graphs have been used to sample from a variety of spaces \citep{Fosdick-18}. MCMC methods are also widely used to sample exponential random graphs \citep{Lusher-12}, and there has been much focus on generating networks that have a desired degree sequence \citep{Stone-05,Tabourier-11,Cooper-07}. This is achieved through the use of an `edge swaps' proposal distribution that preserves the degree sequence of the network throughout the MCMC process. Much of this work focusses on the configuration model; that is, the uniform sampling of networks with a given degree sequence. These have applications when using the configuration model directly or as null models \citep{Stone-05}. Other works sample uniformly from graphs with power-law distributions in a similar manner \citep{Gkantsidis-03}. Uniform sampling can be useful in some situations; however, we are often interested in sampling from a model with a more complicated underlying distribution, and in ensuring we do not oversample rare graphs. Therefore, here we focus on sampling from spaces of graphs that have a non-uniform distribution. Recently, the `edge-switch' proposal in MCMC methods have been used to sample bipartite graphs with only expected degrees that provide a framework to study partially observed networks \citep{Rechner-17}, and the extension of the double swap to a triple swap to allow sampling of `loopy' graphs \citep{Nishimura-18}. Another related work, \citep{Ying-09}, uses link switches to generate synthetic networks preserving properties of a real graph input with privacy and significance testing applications.

\subsection{Connectedness}
We present our algorithm in the context of generating connected random networks. The property of connectedness is often observed in physical networks, such as a telecommunications network, where there is the requirement that a path exist between all nodes. Other physical examples include the Internet routing network. It is also important in the application of social networks. In general each individual may not be connected to all others through some path. However, in the application of epidemics and information diffusion there is particular interested in the network over which information propagates. To participate in a cascade the individual must have come into contact with the contagion; therefore, there is necessarily a path between all individuals in the network over which the cascade is observed.

Many random graph generators do not consider connectivity and simply take the giant component of the resulting graphs or prove properties like the distribution of connected component size in the asymptotic limit. However, in many applications we are interested in generating connected networks of fixed size from our distribution.

Rejection sampling is commonly used to generated networks that display a desired property by simply rejected graphs that do not display this property. While appropriate in some cases, there are many situations in which this method is extremely slow. For example, rejection sampling of simple graphs from the configuration model may be exponential in the size of the graph for some degree sequences \citep{Fosdick-18}. For connectedness, the probability of all nodes being connected can be very low even for quite reasonable parameter values, and so rejection sampling is often not practical. While the probability of connectedness has not been found analytically for Waxman graphs, simple simulations can show that connected graphs are often unlikely. Figure \ref{Waxman_connectedness} shows the proportion of Waxman graphs that are connected after 200 samples, for a variety of parameters, and we can see that as the dependence on distance becomes stronger ($s$ parameter increasing) the probability of connectedness decreases. Additionally, the traditional $\mathcal{O}(N^2)$ sampler makes running even a few hundred samples of the Waxman expensive.

Markov chain methods have been used to produce connected random networks with a prescribed degree sequence \citep{Viger-05,Fosdick-18}, with a particular focus on with a networks in peer-to-peer applications \citep{Cooper-07}. 
\begin{figure}
\centering
\includegraphics[width=0.8\columnwidth]{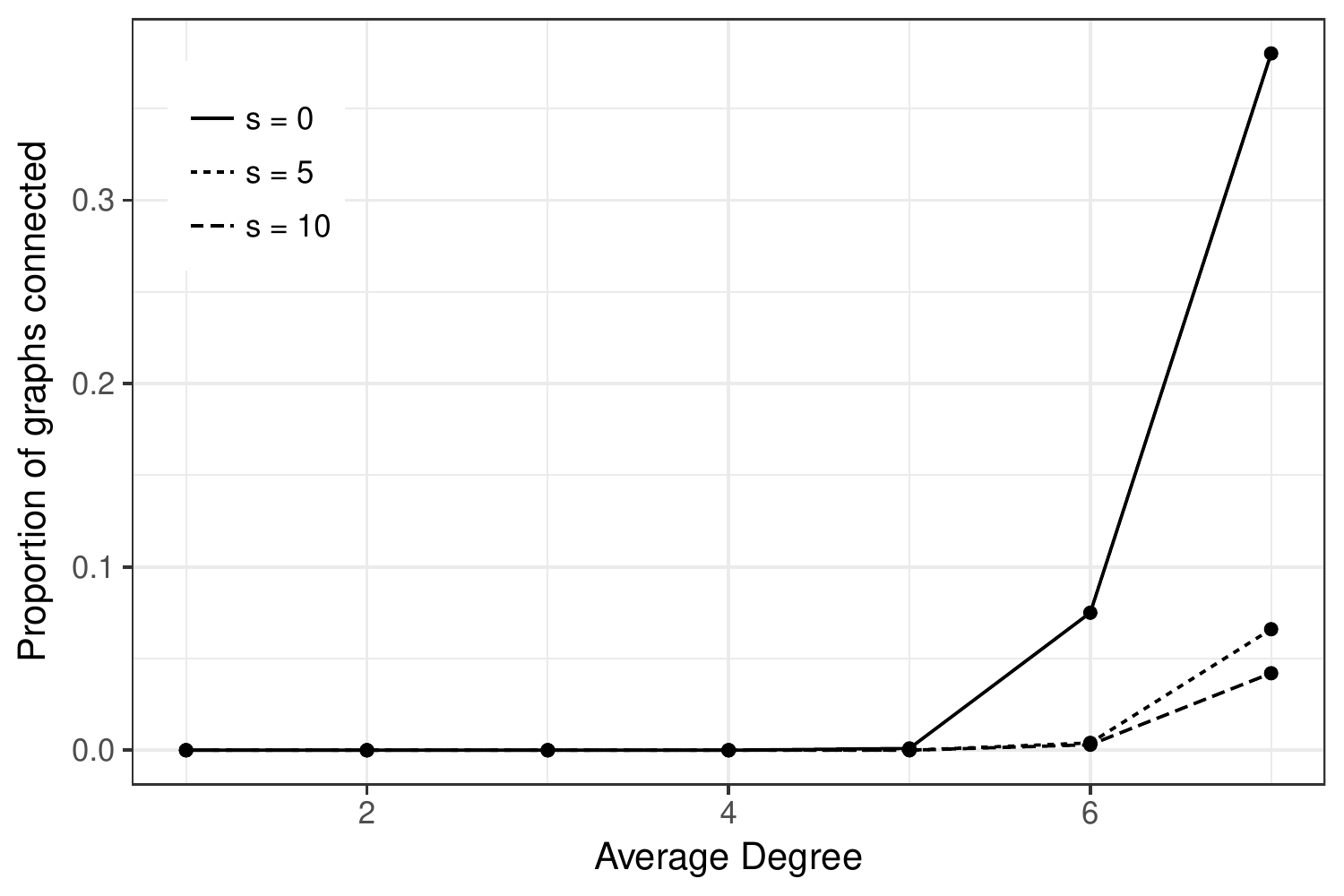}
\caption{Proportion of connected networks in 1000 samples of a Waxman network with $N = 1000$.}
\label{Waxman_connectedness}
\end{figure}

\section{Generating Connected Graphs}
We assume a random graph model that generates an ensemble of sometimes unconnected graphs, and that the model provides a probability distribution across the ensemble, \ie the probability $P(G)$ for each graph $G$. Even if we assume that this probability is calculable, direct simulation from the distribution is usually intractable due to the size of the ensemble. Usually, there is an algorithm to generate graphs from the ensemble.

Given the model, we would like to generate connected graphs with the same conditional probability distribution as the model of interest, \ie we would like to generate connected graphs $G$ with probabilities 
\begin{equation*}
P\{G | G \mbox{ is connected} \}
    = \frac{P\{ G \mbox{ and } G \mbox{ is connected} \}}{P\{ G \mbox{ is connected}\}},
\end{equation*}  
where the numerator is given by:
\begin{equation*}
P\{ G \mbox{ and } G \mbox{ is connected} \}
     = \left\{
       \begin{array}{ll}
          P(G), & \mbox{ for $G$ connected } \\
          0,    & \mbox{ otherwise}.
       \end{array}
     \right.
\end{equation*}

The required connected random graphs are samples from the unknown conditional probability distribution $P\{ G | G \mbox{ is connected}\}$. This leads naturally to the use of well known MCMC methods as the basis for the sampling algorithm.

We implement the Metropolis-Hastings method to generate a Markov chain that will result in samples from the desired distribution. The algorithm produces a new graph $G' = (V, E')$ based on the old graph $G$. The two main components are a symmetric proposal distribution that can explore the entire space and a tractable acceptance ratio. 

We initialise the algorithm using the underlying model to create a random graph, $G^{(-1)}$, with $P(G^{(-1)})>0$. This network is connected by adding arbitrary links. The graph need not be necessarily chosen with the correct probability, so in this case almost any procedure to obtain connectivity is adequate. Whichever connectivity procedure is used leads to a connected random graph $G^{(0)}$ used as the input to the M-H algorithm.

The process described in detail below. 

\begin{new_algorithm}[t]
  \begin{algorithmic}[1]  
    \STATE Generate $G^{(-1)}$ from the model
    \STATE Connect $G^{(-1)}$ to get $G^{(0)}$
    \FOR{k=1..K}
       \STATE Generate a random edge $(i,j)$
       \IF{$(i,j) \in E$}
       	  \STATE Remove the edge:   $E' = E \, \backslash (i,j)$
       	  \IF {$G'$ is connected}
          \STATE accept $G'$ with probability $P(G')/P(G)$
          \ELSE 
          \STATE reject $G'$
          \ENDIF
       \ELSE
          \STATE Add edge: $E' = E \cup (i,j)$
          \STATE accept $G'$ with probability $P(G')/P(G)$
       \ENDIF
    \ENDFOR
  \end{algorithmic}
      \caption{Metropolis-Hastings method for generating connected graphs.}
    \label{alg:mh_general}
\end{new_algorithm}
\textbf{Step 1 - Proposal: } The probability density $Q(G' | G)$, is the proposal distribution that gives the next candidate for the algorithm. An advantageous feature of $Q$ for the M-H algorithm is that it be symmetric, \ie $Q(G'|G) = Q(G|G')$, as this simplifies the acceptance ratio. 
  
Here we perform the algorithm link by link. At each step, we select a node pair $(i,j)$ at random, and consider adding or removing a link to obtain the new network. In practice we choose two distinct nodes at random and consider the possible link between them.
  
  Mathematically,
  \begin{enumerate}
  \item if $(i,j) \in E$ then $E' = E \, \backslash (i,j)$,
  \item if $(i,j) \not\in E$ then $E' = E \cup (i,j)$.
  \end{enumerate}

All node pairs are chosen with equal probability, so $Q(G'|G) = 1/(N(N-1))$ for all $G$ and $G'$ that differ by one link. Therefore, the transition is symmetric.

This proposal has been used in graph sampling previously, notably in applications related to sampling exponential random graphs, \eg \citep{Lusher-12}, and there is no consideration of connectivity in this step. 

\textbf{Step 2 - Acceptance: } The Metropolis-Hastings acceptance ratio (the probability of accepting the proposed transition) given that the proposal is symmetric is given by
\begin{equation}
 \alpha = \min \left\{ 1, \frac{P\{G' | G' \mbox{ is connected} \}}{P\{G | G \mbox{ is connected} \}} \right\}.
 \label{eq:acceptance_long}
 \end{equation}

If the proposed graph has a higher probability than the previous graph we accept the move. If not, we accept with some probability dependent on the ratio of the two graph probabilities. However, the ratio is intractable in this form, as we cannot calculate $P\{G | G \mbox{ is connected} \}$.

To determine a tractable acceptance ratio, we consider the connectivity of each proposed graph. Recall, we start with a valid connected graph $G^{(0)}$. 
If $G'$ is unconnected, then $P\{G' | G' \mbox{ is connected} \} = 0$, so unconnected graphs will never be accepted; therefore, we remain in the space of connected graphs.

We use this to establish a tractable ratio. When $G$ and $G'$ are connected, the conditionals can be dropped from the probabilities, as $P\{ G \mbox{ is connected}\}$ is constant over the ensemble.

This gives
\begin{equation}
\label{eq:acceptance_prob} 
 \alpha = \min \left\{ 1, \frac{P(G' )}{P(G )} \right\},
\end{equation}
for connected graphs $G$ and $G'$. The ratio is tractable in many cases where we can calculate the ratio of the probability distributions. If all edges are independent then this can be calculated very quickly.

The process is iterated a number of times until the Markov chain converges and the networks are being sampled from the stationary distribution of interest. 

To implement this algorithm we must check the connectivity of the graph when removing a link. There are a variety of algorithms for checking connectivity \citep{Eppstein-97}. We use a simple breadth first search with complexity $\mathcal{O}(N+|E|)$, as we are interested in sparse graphs with $|E| \sim \mathcal{O}(N)$, meaning the search is $\mathcal{O}(N)$. After removing a link $(i,j)$ the graph remains connected if and only if a path still remains between $i$ and $j$. Therefore, determining if the graph still has a path between $i$ and $j$, although still $\mathcal{O}(N)$, is likely to be faster than the worst case, especially on spatial graphs.

This algorithm will work well on networks where each edge exists independently of any other, \eg the GER graph, inhomogenous random graphs \citep{Bollobas-07}, or SERNs. In these cases the calculation of $P(G')/P(G)$ is a simple ratio of edge probabilities. In principle this algorithm can be applied to any model in which every graph has positive probability prior to adding that extra constraint, although $P(G')/P(G)$ may be hard to calculate. Note also that the algorithm, as described here, will work only for graph models that assign positive probability to graphs with a different number of edges. For example, the configuration model network has a fixed degree sequence, hence a fixed number of edges. Therefore, the proposal of adding or removing a single edge will be inappropriate as it will break the degree sequence. A simple change of proposal distribution allows for sampling from these networks \citep{Viger-05,Fosdick-18}

\section{Theoretical Convergence}
\label{theoretical_convergence}
\begin{theorem}
\label{Thm:convergence}
Algorithm \ref{alg:mh_general} generates samples from the random graph ensemble with probability distribution $P\{G | G \mbox{ is connected} \}$.
\end{theorem}
\begin{proof}
Theorem (7.4) of Robert \& Casella \citep{Robert-04} states that the chain produced by the Metropolis-Hastings algorithm (Algorithm \ref{MH_algorithm}) converges to the stationary distribution $\pi$ if:
\begin{enumerate}
\item it is irreducible, and
\item it is aperiodic.
\end{enumerate}

Consider the Markov chain produced by Algorithm \ref{alg:mh_general}. We show there exists a sequence of a finite number of steps with positive probability from any connected graph $H$ to any connected $H'$, \ie~$P(H~\rightarrow~H')>0$.
We must ensure that the graph remains connected in all steps. Therefore, consider adding all edges not in $H$ to create a clique. Then remove the edges in subsequent steps to reach~$H'$.

\begin{equation*}
P(H \rightarrow H_{\mathrm{clique}} \rightarrow H') = P(H \rightarrow H_{\mathrm{clique}} ) P(H_{\mathrm{clique}} \rightarrow  H').
\end{equation*}
If every connected graph in the ensemble has non-zero probability, both terms on the RHS have positive probability. Therefore, the chain is \textbf{irreducible}.

A sufficient condition for for the Markov chain to be aperiodic is to choose $Q$ such that the probability of the event $\left\{X^{(t+1)} = X^{(t)}\right\}$ is non-zero for some state. If the removal of an edge destroys connectivity the transition is rejected and the chain remains in the current state. Therefore, the chain is \textbf{aperiodic}.

Note that the acceptance probability construction ensures $\pi = P\{G | G \mbox{ is connected} \}$. Hence, by Theorem (7.4) of Robert \& Casella, Algorithm \ref{alg:mh_general}, with acceptance probability $\alpha$ (\ref{eq:acceptance_prob}) converges to the distribution of interest.

\end{proof}

Unfortunately this result only assures us that after infinite time the process will be sampling from the distribution of interest. We show evidence for convergence in finite time in Section~\ref{sec:convergence}.

\section{Complexity}
\label{complexity}
\begin{theorem}
Algorithm \ref{alg:mh_general} with $K$ iterations has computational complexity $\mathcal{O}(K)$, independent of the size of the graph, for sparse graphs.
\end{theorem}
\begin{proof}

We use a neighbourhood list stored in a hash map to describe the edges in the graph. This results in expectedly $\mathcal{O}(1)$ operations to check edges for existence and add/remove edges at each iteration. We check for connectivity when edge removal is proposed. The breadth first search algorithm is $\mathcal{O}(N)$ for a sparse network with $N$ nodes. For a sparse graph the number of edges is $\mathcal{O}(N)$, and so the probability of selecting an edge to delete is $\mathcal{O}(1/N)$. That is, for large $N$
\begin{equation*}
P\{\text{edge }(i,j)\text{ exists}\} \sim \frac{1}{N}.
\end{equation*}


So, the probability there exists an edge between the two chosen nodes, requiring the $\mathcal{O}(N)$ connectedness routine, decreases like $1/N$. Therefore, each iteration is on average $\mathcal{O}(1)$, and overall the algorithm is  $\mathcal{O}(K)$ in the number of iterations.

\end{proof}

%
%

\section{SERN example}
Here we present the example of spatially embedded networks to demonstrate the algorithm.

Edges in a SERN are independent (conditional on distance), and hence the probability distribution of a spatially embedded random network is given by
\begin{equation}
   P(G) = \prod_{(i,j) \in E} p_{ij}   \prod_{(i,j) \not\in E} (1 - p_{ij}) ,
   \label{eq:prob_dist_SERN}
\end{equation}
where $p_{ij}$ is the probability of an edge for the specific SERN of interest. For example, in the case of a Waxman network the edge $(i,j)$ is given by
\begin{equation*}
p_{ij} = qe^{-sd_{ij}},
\end{equation*}
for nodes separated by distance $d$. In the Waxman formulation, $d$ is calculated by the Euclidean distance.

Using~(\ref{eq:prob_dist_SERN}) above, the acceptance probability when adding a link $(i, j)$ becomes
\begin{equation*}
\frac{P(G')}{P(G)} = \frac{p_{ij}}{1-p_{ij}},
\end{equation*}
and for removing a link is
\begin{equation*}
\frac{P(G')}{P(G)} = \frac{1-p_{ij}}{p_{ij}}.
\end{equation*}

While the probability distribution of the ensemble is known, it is often difficult in practice to determine the value of $P(G)$ explicitly. Here, we only require the ratio of the probabilities between each pair of graphs, a much easier calculation.

Often, we assume that we are dealing with sparse graphs. Dense graphs are more likely to be connected, and so would not require this algorithm. Additionally, in physical networks there is a cost constraint of constructing links, resulting in many sparse real-world networks. 

\subsection{Single link Markov chain: Waxman}
\label{relaxation}
Theorem \ref{Thm:convergence} guarantees convergence in infinite time; however, to be practical we would like it to mix in a reasonable number of steps. We would like estimate the number of iterations $K$ required to have evidence that the Markov chain has sufficiently converged to the stationary distribution. This will depend on the number of nodes in the graph, and, for now, we assume independence between edges.



Our method is performed by proposing a change to a single node pair in each step. Therefore, let us consider that we choose a node pair $(i,j)$ in the graph $G$ with probability $\delta$. In this case we choose all node pairs with equal probability, \ie $\delta=1/(N(N-1))$. While we will analyse one node pair, by choosing a link in the graph with probability $\delta$ we are considering the graph as a whole.

Figure \ref{figure of markov chain} shows the transition probabilities of one node pair. Note that the probability of remaining in the state is through two processes; either not choosing the node pair or choosing the link and not accepting the change. \ie $ 1-\delta + \delta(1-\alpha) =1-\delta\alpha$. For each of the node pairs the probability of accepting a change is

\begin{eqnarray*}
\alpha_{\text{a}}^{ij} = \text{min }\left(1, \frac{p_{ij}}{1-p_{ij}}\right) & \text{if adding}, \\
\alpha_{\text{r}}^{ij} = \text{min }\left(1, \frac{1-p_{ij}}{p_{ij}}\right) & \text{ if removing.} 
\end{eqnarray*}
\begin{figure}
\begin{center}
\includegraphics[scale=0.3]{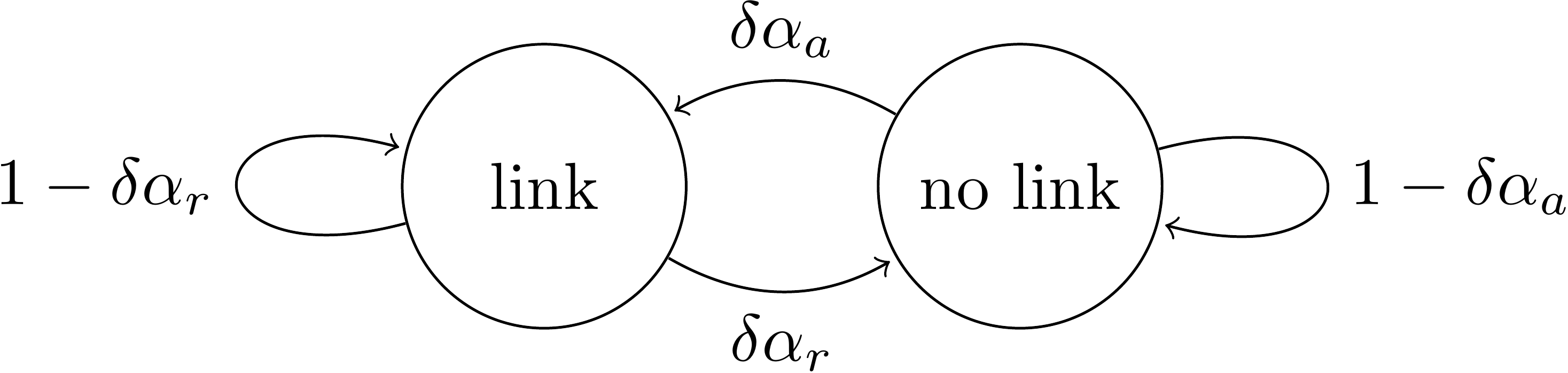}
\end{center}
\caption{Single link in the Markov Chain.}
\label{figure of markov chain}
\end{figure}

Combining the probability of choosing the node pair $(i,j)$ and the transition probabilities, the transition matrix of node pair $(i,j)$ is
\begin{center}
\begin{equation}
P^{ij} = 
\begin{blockarray}{ccc}
\text{link} & \text{no link} \\
\begin{block}{(cc)c}
1-\delta\alpha_r^{ij} & \delta\alpha_r^{ij} & \text{link}\\
\delta\alpha_a^{ij}  & 1-\delta\alpha_a^{ij} & \text{  no link} \\
\end{block}
\end{blockarray}
.
\label{eq:transition}
\end{equation}
\end{center}

In the limit this converges to the stationary probability of a link between nodes $i$ and $j$
\begin{eqnarray*}
p(\text{link}) &=& \frac{\alpha_a}{\alpha_a + \alpha_r},\\
& = & \frac{\text{min }\left(1, \frac{p_{ij}}{1-p_{ij}}\right)}{\text{min }\left(1, \frac{p_{ij}}{1-p_{ij}}\right) + \text{min }\left(1, \frac{1-p_{ij}}{p_{ij}}\right)},\\
& = &  \left\{
       \begin{array}{ll}
          \frac{\frac{p_{ij}}{1-p_{ij}}}{\frac{p_{ij}}{1-p_{ij}} + 1}, & \mbox{ if $p < 0.5$}, \\
          \frac{1}{1 + \frac{1-p_{ij}}{p_{ij}}}    & \mbox{ if $p \geq 0.5$},
       \end{array} 
       \right. \\
 & = & p_{ij} .
\end{eqnarray*}

Hence, the MCMC process will produce a network with the required link probability.

To extend this to the connected case of sampling from $P\{G | G \mbox{ is connected} \}$ we note that the probability of removing a link and moving into a `no link' state where the network is disconnected is zero. As we never start in this absorbing state (the initial network is always connected), the connected system of interest is equivalent to the simplified case presented above.

The mixing of the Markov chain is important in the application of the algorithm in finite time. The spectral gap controls the rate of exponential decay to equilibrium and the relaxation time gives an indication of how fast the chain converges. The two eigenvalues of the transition matrix~(\ref{eq:transition}) are $\lambda_1 = 1$ and $\lambda_2 = 1-\delta\alpha_r - \delta\alpha_a$, giving a spectral gap of $\gamma^* = \delta(\alpha_r + \alpha_a)$. Note that we select edge $(i,j)$ with probability $\delta\sim1/N^2$ and $\alpha_r + \alpha_a$ is constant for any given link.

The relaxation time is given by,
\begin{eqnarray*}
t_{\mathrm{rel}} &=& \frac{1}{\gamma^*}, \\
& = &  \frac{1}{\delta(\alpha_r + \alpha_a)},
\end{eqnarray*}

where
\begin{eqnarray*}
\alpha_r + \alpha_a &=&  \left\{
       \begin{array}{ll}
          \frac{1}{1-p_{ij}}, & \mbox{ if $p < 0.5$}, \\
          \frac{1}{p_{ij}}    & \mbox{ if $p \geq 0.5$},
       \end{array} 
       \right. \\
       & \in & [1,2].
\end{eqnarray*}
 
Therefore, in general, $t_{\mathrm{rel}} \sim N^2$ for the graph, and we expect that $K \sim \mathcal{O}(N^2)$ for the algorithm to converge.

Above we assume that node pair transitions are independent. However, when consider connectedness, the presence or absence of other edges may prevent a particular edge being removed. This will increase the mixing time of the chain as it is possible that the most likely path from one graph to another travels through some unconnected graph. Nevertheless, this analysis gives us a lower bound on and an intuition about the mixing time of our algorithm. To investigate the real mixing time we next turn to the practicalities of implementing the algorithm and investigate the convergence.

\section{Implementation}
\label{sec:convergence}
Section~\ref{theoretical_convergence} showed that Algorithm \ref{alg:mh_general} will converge to the distribution of connected Waxman graphs in infinite time, but we expect approximate convergence in $K \sim \mathcal{O}(N^2)$ steps. The critical question becomes, how long is required in practice?

We implement the algorithm described above using the NetworkX package in Python 2.7.13 \citep{Python} to produce connected SERNs.

In order to simulate networks in finite time we must provide evidence for the convergence of the chain. Many applications of MCMC use visual means to determine when the chain seems to have converged. Here we use a heuristic that uses statistics of the graph. 

Summary statistics are often used to describe network ensembles. Here we utilise the distributions of two summary statistics over the ensemble to determine convergence, the distribution of average degrees and average path length. When we condition on connectedness, we expect a slight increase in average degree to allow for connectedness. This results in a shift in the distribution of the average degree over the ensemble. Conversely, we expect the average path length to decrease as the starting graph $G^{(0)}$ will have longer links than a typical Waxman graph (as we added random links to connect the graph). Note that the average edge length has particular significance in SERNs, and with the average number of edges (closely related to average degree) creates a minimal set of sufficient statistics for the parameters of the Waxman graph \citep{Roughan-15}. After convergence we expect no change in the distribution of summary statistics of the network as they are being drawn from the same underlying distribution. We investigate the change in these statistics to provide evidence for convergence.
\begin{figure}
\centering
\includegraphics[width=0.8\columnwidth]{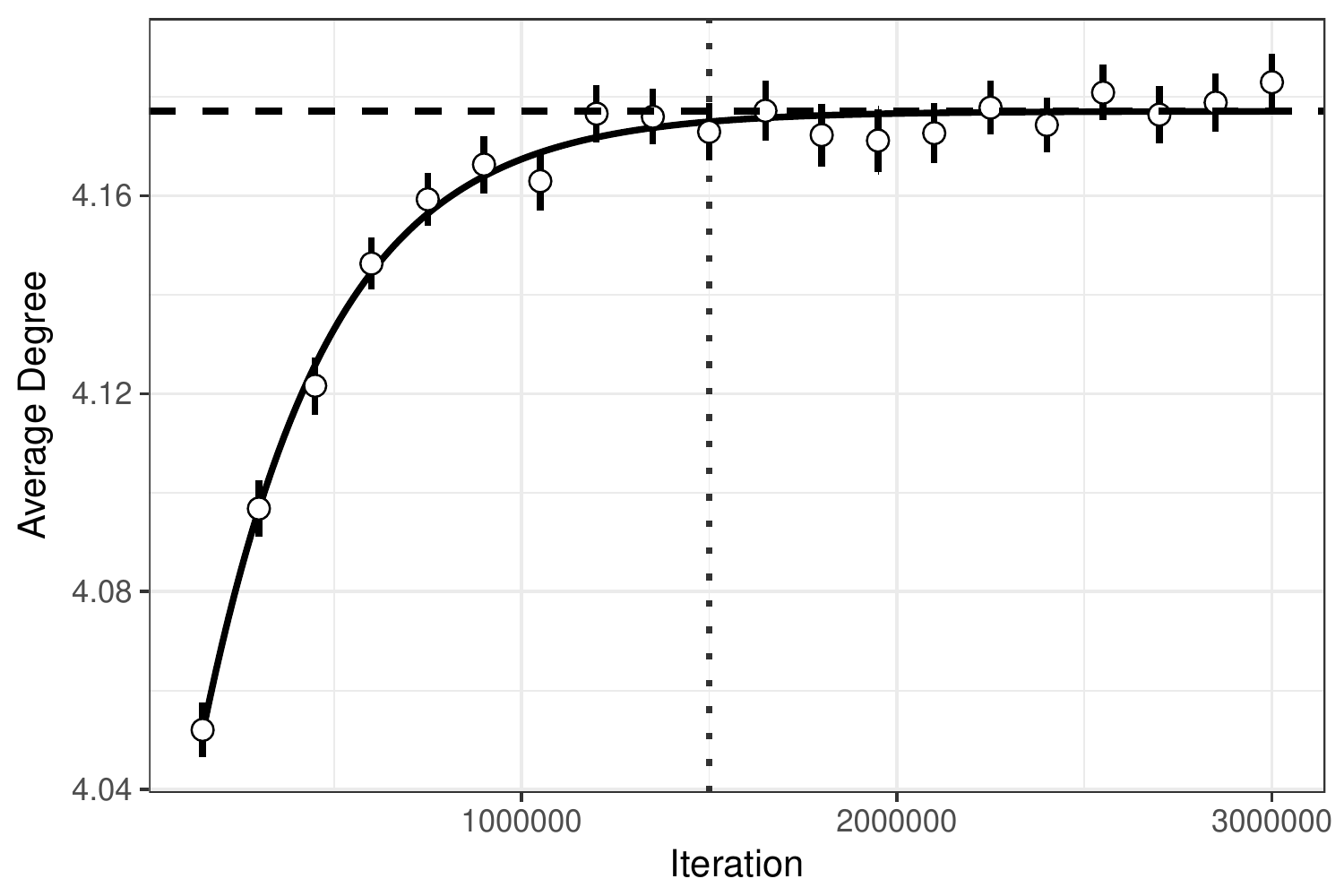}
\caption{Average degree over the MCMC process for a Waxman network with $N=1000$ nodes. The means of 200 runs are shown with 95\% confidence intervals. The solid fitted regression curve is shown, and the dashed line represents the fitted asymptote. Note that there is evidence for convergence at approximately 1.5 million iterations}
\label{fig:ave_degree_change}
\end{figure}

\begin{figure}
\centering
\includegraphics[width=0.8\columnwidth]{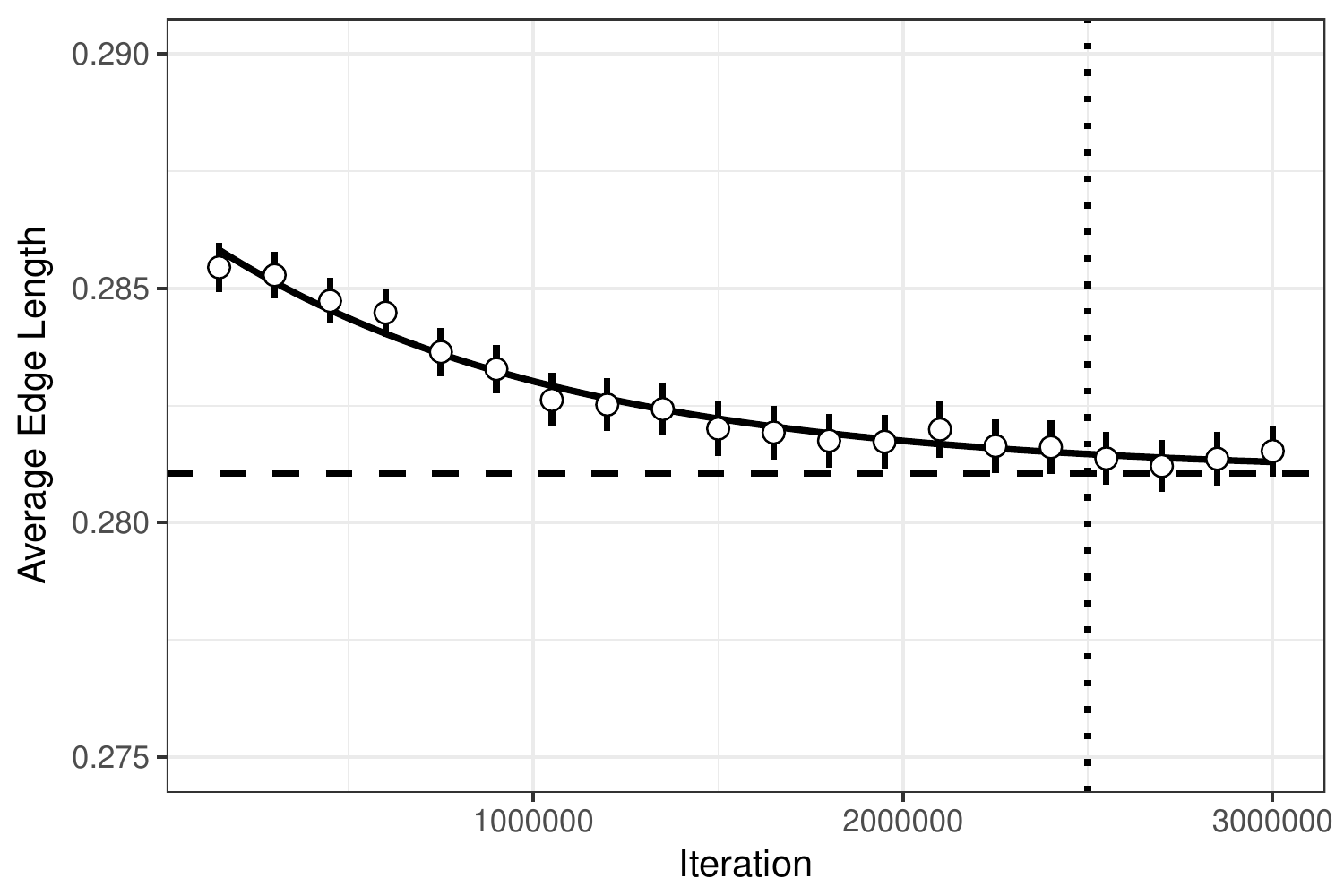}
\caption{Average edge length over the MCMC process for a Waxman network with $N=1000$ nodes. The means of 200 runs are shown with 95\% confidence intervals. The solid fitted regression curve is shown, and the dashed line represents the fitted asymptote. Note that there is evidence for convergence at approximately 2.5 million iterations}
\label{fig:ave_edge_change}
\end{figure}

Figure \ref{fig:ave_degree_change} shows the confidence intervals of average degree in 200 chains of the MCMC process, \ie values at intervals along the process in 200 runs of the algorithm. This demonstrates a steady increase in average degree as the algorithm progresses. We suggest that there is no significant change in average degree after 1.5 million iterations, and we have reached the average degree of $P\{G | G \mbox{ is connected} \}$. The average edge length, Figure \ref{fig:ave_edge_change} changes significantly but the magnitude of the change is much smaller. Additionally, it appears to converge slightly slower than the average degree, reaching within 99.9\% of the fitted asymptote at $\sim$2.5 million iterations. Therefore, we have evidence that the system has converged and we are sampling from the posterior distribution of connected Waxman networks.

\subsection{Iterations until convergence}
To estimate $K$, the number of steps required until convergence, we must investigate how the number of iterations to convergence scales with the number of nodes in the network. Therefore, determining convergence by eye is insufficient. We develop a framework to automate the process and give estimates of the required iterations to convergence.  
First, we use non-linear least squares in R \citep{R} to fit an exponential function to the average degree over the iterations and determine when the average degree distribution is no longer changing. The function, of the form
\begin{equation*}
f(x) = C + Ae^{-Bx},
\end{equation*}  
is fitted to the full data (not just the means) to determine the parameter $C$. This fitted parameter is used as the average degree of the target ensemble $P\{G | G \mbox{ is connected} \}$ after convergence, see Figure \ref{fig:ave_degree_change}. We define strong evidence for convergence to be when the fitted values are within 0.1\% of this value.
%

We apply this framework to the MCMC process for varying $N$ to determine the scaling of convergence. From the results in Figure \ref{fig: fitted_loglog} we note that the line of best fit is a power-law with an exponent of $1.99 \pm 0.04$. We conclude that the mixing time of this algorithm (number of iterations to convergence) is approximately $\mathcal{O}(N^2)$. This agrees with the theoretical analysis in Section~\ref{relaxation}. We note that we see the same results when fitting other functions, for example a logistic curve.

To provide further evidence for this $\mathcal{O}(N^2)$ complexity we conduct a similar analysis with the average edge length. We again fit an exponential model and as before the parameter C is the asymptote taken to be the average edge length of the target ensemble. The iterations until convergence, as calculated by the average edge length is shown in Figure \ref{fig: fitted_loglog} (triangles). The average edge length converges more slowly than the average degree. This is expected as some links cannot be removed until other links provide new paths through the network. It displays the same scaling, with the exponent of the line of best fit of $2.01\pm 0.06$, providing further evidence for convergence in approximately $\mathcal{O}(N^2)$.

Combining with results from Section~\ref{complexity}, the overall complexity of the algorithm is $\mathcal{O}(N^2)$.


\begin{figure}
\centering
\includegraphics[width=0.8\columnwidth]{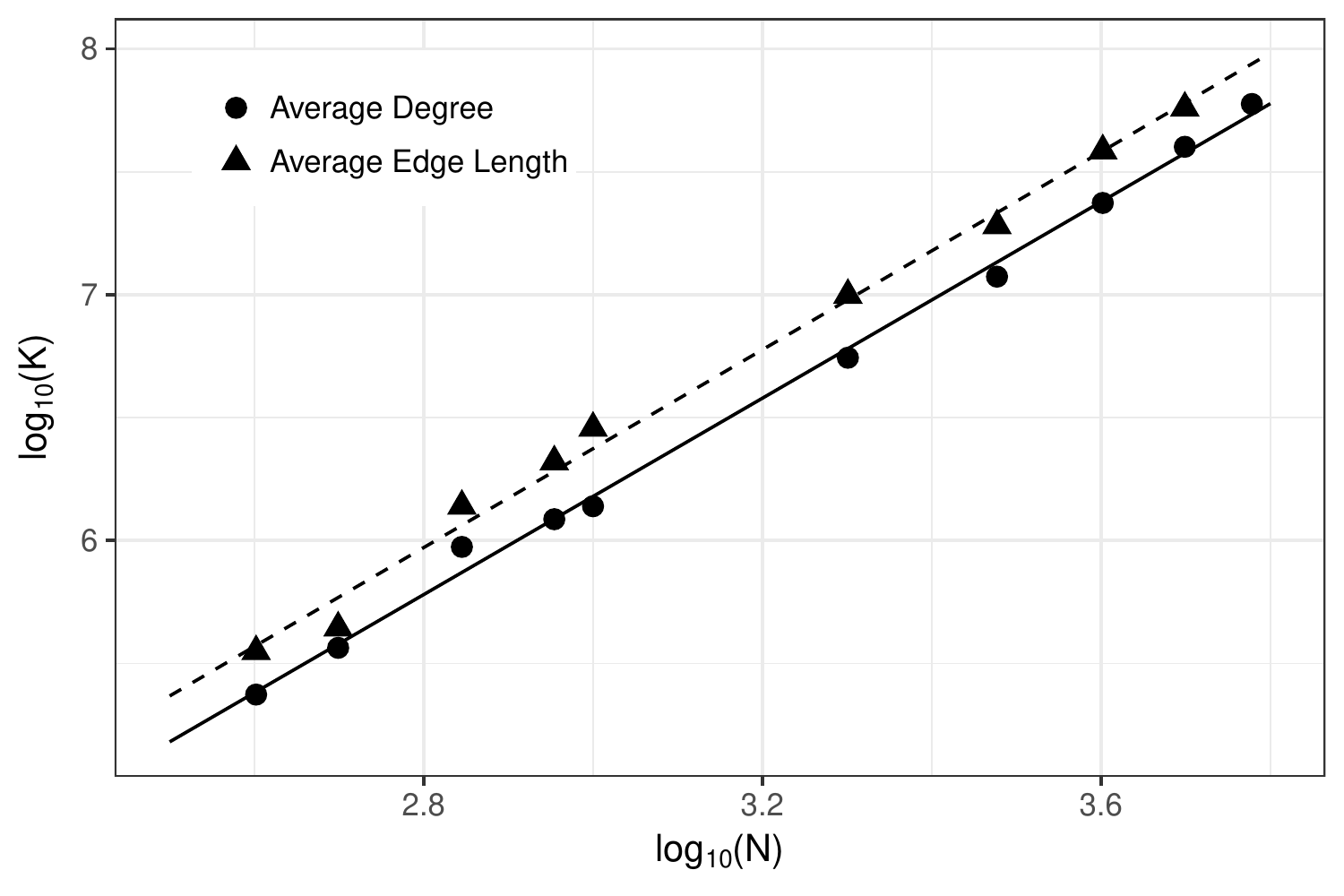}
\caption{Log-log plot of the iterations to convergence of the algorithm for varying size networks using average degree (circles) and average edge length (triangles) as the summary statistic. The slope of the fitted line for average degree (solid) is $1.99  \pm 0.04$, and for average edge length (dashed) is $2.01 \pm 0.06$. This supports the $\mathcal{O}(N^2)$ mixing time expected over the edges in a network.}
\label{fig: fitted_loglog}
\end{figure}

\section{Discussion}
We have introduced our algorithm in the context of generating connected networks. However, this method generalises to generate networks from the probability distribution given by
\begin{equation*}
P(G \vert G \mbox{ has some properties}),
\end{equation*}
assuming the properties can be tested. For example, generating a network without self loops or multi-edges would be easily implemented as above. Although we only condition on connectedness here, the process is not restricted to a single property, a set of properties can be used. 

It is worth noting that the proposal distribution can affect the properties that can be tested, $Q(G'|G)$. In this implementation, the proposal considers individual node pairs, and each step changes a single link. If we were to fix the number of triangles or exact degree sequence (\ie the configuration model) in the network, our proposal distribution would need to facilitate this. In these cases, an `edge swap' proposal in which the number of links remains constant, for example \citep{Gkantsidis-03, Viger-05}, would be an appropriate choice. There are many other constraints this method could be applied to in this form, or by changing the proposal. Other types of connectivity ($k$-connectivity) and using other models are natural extensions, and this type of method has applications in modelling many real world networks, for example ancestries where the relationships between animals must satisfy a variety of conditions.

The above algorithm assumes that the probability distribution of the network has the form in~(\ref{eq:prob_dist_SERN}). However, other probability distributions, for example that of exponential random graphs \citep{Lusher-12} can easily be used. Note that we must be able to calculate the ratio of densities of graphs that differ by one link.

We initialise the algorithm by simulating a graph from the model of interest; \eg the Waxman network, and connecting arbitrarily. However, any connected network can be used in this step as the MCMC process by design forgets the initial point of the Markov chain. This is particularly useful where the generation of the graph of interest is computationally expensive. However, starting `further' from the distribution of interest may increase time to convergence.

We have also focussed on the simulation of a single graph, assuming that multiple graphs can be sampled by running multiple instances. However, we can sample multiple graphs from the same chain. Thinning of the chain will need to be employed to create independent samples. We expect number of iterations until independent samples to be of the same order (not necessarily the same time) of mixing time, $\mathcal{O}(N^2)$. This is intuitive as each node pair must have the opportunity to change to create independent graphs.

A speed up heuristic, proposed by Gkantsidis \textit{et al.}~\citep{Gkantsidis-03} on a simple Markov Chain, attempts to reduce the requirement of checking connectedness by only running the check after $T$ `flip' transitions and rejecting if disconnected. This produces a concatenation of Markov Chains that maintain the required stationary distribution. This speed up factor could easily be applied here to the single link Metropolis-Hastings method. However, we only check for connectedness when the proposal removes a link, compared to every step. Rejecting all $T$ transitions (both link additions and removals) if the graph becomes disconnected would slow mixing. Hence, it is unlikely that this speed up method would produce the same dramatic increase in complexity observed in \citep{Viger-05}. Alternative connectivity algorithms present opportunities for improving complexity. Eppstein \textit{et al.}~\citep{Eppstein-97} present a dynamic connectivity check in $\mathcal{O}(\sqrt{N})$ per change in the graph. This is promising; however this is required at every addition or deletion of an edge, rather than only at deletion, so would not improve overall performance. These dynamic algorithms provide an opportunity to allow sampling of graphs with other properties, \eg $k$-connectivity.
\section{Conclusion \& Future Work}
This paper describes an algorithm to create random networks from a known ensemble conditioned on an extra desired property. We use a Bayesian framework, implemented with MCMC, to generate connected random networks. This implementation can be extended to include other properties of a network. We demonstrate the time complexity is $\mathcal{O}(N^2)$ with strong evidence of convergence to the desired ensemble. Future work includes applying this algorithm to other constraints and networks, and improving the efficiency of the algorithm. Extensions of the Metropolis-Hastings method, such as importance sampling, aim to improve mixing and complexity of convergence that could also be investigated in this context  \citep{Kroese-11}. 

\section*{Funding}
This work was supported by the Data to Decisions CRC (D2D CRC), the Cooperative Research Centres Programme; the ARC Center of Excellence for Mathematical and Statistical Frontiers (ACEMS), and an Australian Government Research Training Program (RTP) Scholarship.

\section{References}

\end{document}